\documentclass[a4paper,UKenglish,cleveref, autoref, thm-restate]{lipics-v2021}
\title{Suffix sorting via matching statistics} %TODO Please add
\author{Zsuzsanna Lipt{\'a}k}{Department of Computer Science, University of Verona, Italy}{zsuzsanna.liptak@univr.it}{https://orcid.org/0000-0002-3233-0691}{}
\author{Francesco Masillo}{Department of Computer Science, University of Verona, Italy}{francesco.masillo@univr.it}{https://orcid.org/0000-0002-2078-6835}{}

\author{Simon J.\ Puglisi}{Helsinki Institute for Information Technology (HIIT) \and Department of Computer Science, University of Helsinki, Finland}{simon.puglisi@helsinki.fi}{https://orcid.org/0000-0001-7668-7636}{}

\authorrunning{Zs. Lipt\'ak, F. Masillo and S.J. Puglisi} %TODO mandatory. First: Use abbreviated first/middle names. Second (only in severe cases): Use first author plus 'et al.'

\Copyright{Zsuzsanna Lipt\'ak, Francesco Masillo and Simon J. Puglisi} %TODO mandatory, please use full first names. LIPIcs license is "CC-BY";  http://creativecommons.org/licenses/by/3.0/



\category{} %optional, e.g. invited paper

\relatedversion{} %optional, e.g. full version hosted on arXiv, HAL, or other respository/website
\supplement{\url{https://github.com/fmasillo/sacamats}}%optional, e.g. related research data, source code, ... hosted on a repository like zenodo, figshare, GitHub, ...

\nolinenumbers %uncomment to disable line numbering

\hideLIPIcs  %uncomment to remove references to LIPIcs series (logo, DOI, ...), e.g. when preparing a pre-final version to be uploaded to arXiv or another public repository

\EventEditors{John Q. Open and Joan R. Access}
\EventNoEds{2}
\EventLongTitle{42nd Conference on Very Important Topics (CVIT 2016)}
\EventShortTitle{CVIT 2016}
\EventAcronym{CVIT}
\EventYear{2016}
\EventDate{December 24--27, 2016}
\EventLocation{Little Whinging, United Kingdom}
\EventLogo{}
\SeriesVolume{42}
\ArticleNo{23}

\usepackage{multicol}
\usepackage{amsthm}
\usepackage{latexsym,amsmath,amssymb,stmaryrd,mathtools}
\usepackage{graphicx,epsfig,tikz}
\usepackage{float}
\usepackage[titlenumbered,ruled]{algorithm2e}
\usetikzlibrary{shapes.geometric,positioning,chains,fit,calc}
\tikzset{near start abs/.style={xshift=1cm},
	node/.style={circle,draw},
	nodeone/.style={circle,draw,gray, line width=0.7mm},
	nodered/.style={circle,draw,red, line width=0.7mm}}

\usepackage[useregional]{datetime2}

\newcommand{\BWT}{\textit{BWT}}
\newcommand{\Oh}{{\cal O}}
\newcommand{\SA}{\textit{SA}}
\newcommand{\GSA}{\textit{GSA}}
\newcommand{\ISA}{\textit{ISA}}
\newcommand{\LCP}{\textit{LCP}}
\newcommand{\PLCP}{\textit{PLCP}}

\newcommand{\MS}{\textit{MS}}

\newcommand{\PSV}{\textit{PSV}}
\newcommand{\NSV}{\textit{NSV}}

\newcommand{\Y}{\textit{Y}}
\newcommand{\Yc}{\textit{Yc}}
\newcommand{\cY}{\textit{cY}}

\newcommand{\pred}{\textit{pred}}
\newcommand{\LCPsum}{\textit{LCPsum}}
\newcommand{\CMS}{\textit{CMS}}
\newcommand{\eCMS}{\textit{eCMS}}

\newcommand{\lcp}{\textit{lcp}}
\newcommand{\ems}{\textit{ems}}
\newcommand{\ihead}{\textit{i-head}}
\newcommand{\sacamats}{{\tt sacamats}}
\newcommand{\sais}{{\tt sais}}
\newcommand{\divsufsort}{{\tt divsufsort}}
\newcommand{\gsufsort}{{\tt gsufsort}}
\newcommand{\bigBWT}{{\tt bigBWT}}

\definecolor{darkred}{RGB}{128,0,0}
\definecolor{darkgreen}{RGB}{0,128,0}
\definecolor{lightgreen}{RGB}{224,255,224}
\definecolor{darkblue}{RGB}{0,0,128}
\definecolor{lightblue}{RGB}{224,244,255}

\renewcommand{\epsilon}{\varepsilon}

\begin{document}

\maketitle

\bigskip \bigskip

\begin{center} 
{\em Accepted at the Workshop on Algorithms in Bioinformatics (WABI 2022), Sept. 5-9, 2022, Potsdam, Germany}
\end{center}

\bigskip

\begin{abstract}
    We introduce a new algorithm for constructing the generalized suffix array of a collection of highly similar strings. As a first step, we construct a compressed representation of the matching statistics of the collection with respect to a reference string. We then use this data structure to distribute suffixes into a partial order, and subsequently to speed up suffix comparisons to complete the generalized suffix array. Our experimental evidence with a prototype implementation (a tool we call {\tt sacamats}) shows that on string collections with highly similar strings we can construct the suffix array in time competitive with or faster than the fastest available methods. Along the way, we describe a heuristic for fast computation of the matching statistics of two strings, which may be of independent interest.
\end{abstract}

\newpage

%%%%%%%%%%%%%%%%%%%%%%%%%%%%%%%%%%%%%%%%%%%%%%%%%%%%%%%%%%%%%%%%%%%%%%%%%%%%%%%%%%%%%%%
%% SECTION 
%%%%%%%%%%%%%%%%%%%%%%%%%%%%%%%%%%%%%%%%%%%%%%%%%%%%%%%%%%%%%%%%%%%%%%%%%%%%%%%%%%%%%%%

\section{Introduction}\label{sec:introduction}

Suffix sorting---the process of ordering all the suffixes of a string into lexicographical order---is the key step in construction of suffix arrays and the Burrows-Wheeler transform, two of the most important structures in text indexing and biological sequence analysis~\cite{OhlebuschBook,MBCT15,AKO04}. As such, algorithms for efficient suffix sorting have been the focus of intense research since the early 1990s~\cite{MM93,PST07}.

With the rise of pangenomics, there is an increased demand for indexes that support fast pattern matching over collections of genomes of individuals of the same species (see, e.g.,~\cite{GNP20,PZ20,ROBLGB22}). With pangenomic collections constantly growing and changing, construction of these indexes---and in particular suffix sorting---is a computational bottleneck in many bioinformatics pipelines. While traditional and well-established suffix sorting tools such as {\tt divsufsort}~\cite{divsufsort,0001K17} and {\tt sais}~\cite{sais-lite,NZC11} can be applied to these collections, specialised algorithms for collections of similar sequences, perhaps most notably the so-called {\tt BigBWT} program~\cite{BoucherGKLMM19}, are beginning to emerge.

In this paper we describe a suffix sorting algorithm specifically targeted to collections of highly similar genomes that makes use of the {\em matching statistics}, a data structure due to Chang and Lawler, originally used in the context of approximate pattern matching~\cite{CL94}. The core device in our suffix sorting algorithm is a novel compressed representation of the matching statistics of every genome in the collection with respect to a designated reference genome, that allows determining the relative order of two arbitrary suffixes, from any of the genomes, efficiently. We use this data structure to drive a suffix sorting algorithm that has a small working set relative to the size of the whole collection, with the aim of increasing locality of memory reference. Experimental results with a prototype implementation show the new approach to be faster or competitive with state-of-the-art methods for suffix array construction, including those targeted at highly repetitive data. We also provide a fast, practical algorithm for matching statistics computation, which is of independent interest. 

The remainder of this paper is structured as follows. The next section sets notation and defines basic concepts. In Section~\ref{sec:ms} we describe a compressed representation of the matching statistics and a fast algorithm for constructing it. Section~\ref{sec:comparing} then describes how to use the compressed matching statistics to determine the relative lexicographic order of two arbitrary suffixes of the collection. Section~\ref{sec:putting}  describes a complete suffix sorting algorithm. We touch on several implementation details in Section~\ref{sec:implementation}, before describing experimental results in Section~\ref{sec:experiments}. Reflections and avenues for future work are then offered.

%%%%%%%%%%%%%%%%%%%%%%%%%%%%%%%%%%%%%%%%%%%%%%%%%%%%%%%%%%%%%%%%%%%%%%%%%%%%%%%%%%%%%%%
%% SECTION 
%%%%%%%%%%%%%%%%%%%%%%%%%%%%%%%%%%%%%%%%%%%%%%%%%%%%%%%%%%%%%%%%%%%%%%%%%%%%%%%%%%%%%%%

\section{Basics}\label{sec:basics}

A string $T$ over an ordered alphabet $\Sigma$, of size $|\Sigma| = \sigma$, is a finite sequence $T=T[1..n]$ of characters from $\Sigma$. We use the notation $T[i]$ for the $i$th character of $T$, $|T|$ for its length $n$, and $T[i..j]$ for the substring $T[i]\cdots T[j]$; if $i>j$ then $T[i..j]=\epsilon$, where $\epsilon$ is the empty string. The substring (or factor) $T[i..]=T[i..n]$ is called the $i$th suffix, and $T[..i]=T[1..i]$ the $i$th prefix of $T$. We assume throughout that the last character of each string is a special character $\$$, not occurring elsewhere in $T$, which is set to be smaller than every character in $\Sigma.$ 

Given a string $T$, the {\em suffix array} $\SA$ is a permutation of the index set $\{1,\ldots,n\}$ defined by: $\SA[i]=j$ if the $j$th suffix of $T$ is the $i$th in lexicographic order among all suffixes of $T$. The {\em inverse suffix array} $\ISA$ is the inverse permutation of $\SA$. The {\em \LCP-array} is given by: $\LCP[1]=0$, and for $i\geq 2$, $\LCP[i]$ is the length of the longest common prefix (lcp) of the two suffixes $T[\SA[i-1]..]$ and $T[\SA[i]..]$ (which are consecutive in lexicographic order).  A variant of the $\LCP$ array is the {\em permuted $\LCP$-array}, $\PLCP$, defined as $\PLCP[i] = \LCP[\ISA[i]]$, i.e.\ the lcp values are stored in text order, rather than in $\SA$ order. We further define $\LCPsum(T)=\sum_{i=1}^{|T|} \LCP[i]$. $\LCPsum$ can be used as a measure of repetitiveness of strings, since the number of distinct substrings of $T$ equals ${(|T|^2+|T|)}/{2} - \LCPsum(T)$. All these arrays can be computed in linear time in $|T|$, see e.g.~\cite{NZC11,KMP09}.

Given the suffix array $\SA$ of $T$ and a substring $U$ of $T$, the indices of all suffixes which have $U$ as prefix appear consecutively in $\SA$. We refer to this interval as {\em $U$-interval}: the $U$-interval is $\SA[s..e]$, where $\{ \SA[s],\SA[s+1],\ldots,\SA[e-1],\SA[e]\}$ are the starting positions of the occurrences of $U$ in $T$.

Let ${\cal C} = \{S_1,\ldots,S_m\}$ be a collection of strings (a set or multiset). The {\em generalized suffix array} $\GSA$ of ${\cal C}$ is defined as $\GSA[i]=(d,j)$ if $S_d[j..]$ is the $i$th suffix in lexicographic order among all suffixes of the strings from ${\cal C}$, where ties are broken by the document index $d$. The \GSA\ can be computed in time $\Oh(N)$, where $N$ is the total length of strings in ${\cal C}$~\cite{OhlebuschBook}.

Let $R$ and $S$ be two strings. The {\em matching statistics of $S$ with respect to $R$} is an array $\MS$ of length $|S|$, defined as follows. Let $U$ be the longest prefix of suffix $S[i..]$ which occurs in $R$ as a substring, where the end-of-string character $\#$ of $R$ is assumed to be different from, and smaller than that of $S$. Then $\MS[i] = (p_i,\ell_i)$, where $p_i=-1$ if $U=\epsilon$, and $p_i$ is an occurrence of $U$ in $R$ otherwise, and $\ell_i = |U|$. (Note that $p_i$ is not unique in general.) We refer to $U$ as the {\em matching factor}, and to the character $c$ immediately following $U$ in $S$ as the {\em mismatch character}, of position $i$.  For a collection ${\cal C} = \{S_1,\ldots,S_m\}$ and a string $R$, the matching statistics of ${\cal C}$ w.r.t.\ $R$ is simply the concatenation of $\MS_i$'s, where $\MS_i$ is the matching statistics of $S_i$ w.r.t.\ $R$. We will discuss matching statistics in more detail in Section~\ref{sec:ms}. 

For an integer array $A$ of length $n$ and an index $i$, the previous and next smaller values, $\PSV$ resp.\ $\NSV$, are defined as $\PSV(A, i) = \max\{i' < i : A[i'] < A[i]\}$ resp.\  $\NSV(A, i) = \min\{i' > i : A[i'] < A[i]\}$. Note that $\PSV$ resp.\ $\NSV$ is not defined for $i = 1$ resp. $i = n$. In $O(n)$ preprocessing of $A$, a data structure of size $n \log_2(3 + 2\sqrt{2}) + o(n)$ bits can be built that supports answering arbitrary \PSV\ and \NSV\ queries in constant time per query~\cite{F11}. 

Let $X$ be a finite set of integers. Given an integer $x$, the predecessor of $x$, $\pred(x)$ is defined as the largest element smaller than $x$, i.e. $\pred_X(x) = \max\{y\in X \mid y \leq x\}$. 
Using the y-fast trie data structure of Willard~\cite{W83} allows answering predecessor queries in $O(\log\log |X|)$ time using $O(|X|)$ space.

We are now ready to state our problem: 

\begin{quote}
    {\bf Problem Statement:} Given a string collection ${\cal C} = \{ S_1,\ldots, S_m\}$ and a reference string $R$, compute the generalized suffix array $\GSA$ of ${\cal C}$. 
\end{quote}

We will denote the length of $R$ by $n$ and the total length of strings in the collection by $N=\sum_{d=1}^m |S_d|$. As before, we assume that the end-of-string character $\#$ of $R$ is strictly smaller than those of the strings in the collection ${\cal C}$. 
We are interested in those cases where $\LCPsum_R$ is small and the strings in ${\cal C}$ are very similar to $R$. If no reference string is given in input, we will take $S_1$ to be the reference string by default. 

%%%%%%%%%%%%%%%%%%%%%%%%%%%%%%%%%%%%%%%%%%%%%%%%%%%%%%%%%%%%%%%%%%%%%%%%%%%%%%%%%%%%%%%
%% SECTION 
%%%%%%%%%%%%%%%%%%%%%%%%%%%%%%%%%%%%%%%%%%%%%%%%%%%%%%%%%%%%%%%%%%%%%%%%%%%%%%%%%%%%%%%

\subsection{Efficient suffix array construction}\label{sec:SACA}

Currently, the best known and conceptually simplest linear-time suffix array construction algorithm is the SAIS algorithm by Nong et al.~\cite{NZC11}. It cleverly combines, and further develops, several ideas used by previous suffix array construction algorithms, among these {\em induced sorting}, and use of a so-called {\em type array}, already used in~\cite{IT99,KoA05} (see also~\cite{PST07}).

Nong et al.'s approach can be summarized as follows: assign a type to each suffix, sort a specific subset of suffixes, and compute the complete suffix array by inducing the order of the remaining suffixes from the sorted subset. There are three types of suffixes, one of which  constitutes the subset to be sorted first. 

The definition of types is as follows (originally from~\cite{KoA05}, extended in~\cite{NZC11}): Suffix $i$ is {\em $S$-type} (smaller) if $T[i..] < T[i+1..]$, and {\em $L$-type} (larger) if $T[i..] > T[i+1..]$. An $S$-type suffix is {\em $S^*$-type} if $T[i..]$ is $S$-type and $T[i-1..]$ is $L$-type. It is well known that assigning a type to each suffix can be done with a back-to-front scan of the text in linear time.

Now, if the relative order of the $S^*$-suffixes is known, then that of the remaining suffixes can be induced with two linear scans over the partially filled-in suffix array: the first scan to induce $L$-type suffixes, and the second to induce $S$-type suffixes. For details, see~\cite{NZC11} or~\cite{OhlebuschBook}. 

Another ingredient of SAIS, and of several other suffix array construction algorithms, is what we term the {\em metacharacter method}. Subdivide the string $T$ into overlapping substrings, show that if two suffixes start with the same substring, then their relative order depends only on the remaining part; assign metacharacters to these substrings according to their rank (w.r.t.\ the lexicographic order, or some other order, depending on the algorithm), and define a new string on these metacharacters. Then the relative order of the suffixes of the new string and the corresponding suffixes starting with these specific substrings will coincide. In SAIS~\cite{NZC11}, so-called LMS-substrings are used, while a similar method is applied in prefix-free-parsing (PFP)~\cite{BoucherGKLMM19}. Here we will apply this method using substrings starting in special positions which we term insert-heads, see Sections~\ref{sec:comparing} and~\ref{sec:putting} for details.

%%%%%%%%%%%%%%%%%%%%%%%%%%%%%%%%%%%%%%%%%%%%%%%%%%%%%%%%%%%%%%%%%%%%%%%%%%%%%%%%%%%%%%%
%% SECTION 
%%%%%%%%%%%%%%%%%%%%%%%%%%%%%%%%%%%%%%%%%%%%%%%%%%%%%%%%%%%%%%%%%%%%%%%%%%%%%%%%%%%%%%%

\section{Compressed matching statistics}\label{sec:ms}

Let $R,S$ be two strings over $\Sigma$ and $\MS$ be the matching statistics of $S$ w.r.t.\ $R$. Let $\MS[i] = (p_i,\ell_i)$. It is a well known fact that if $\ell_i>0$, then $\ell_{i+1} \geq \ell_i-1$. This can be seen as follows. Let $U$ be the matching factor of position $i$, and $p_i$ an occurrence of $U$ in $R$. Then $U' = U[2..\ell_i]$ is a prefix of $S[i+1..]$ of length $\ell_i-1$, which occurs in position $p_i+1$ of $R$. 

Let us call a position $j$ a {\em head} if $\ell_j > \ell_{j-1}-1$, and a sequence of the form $(x,x-1,x-2,\ldots)$, of length at most $x-1$, a {\em decrement run}, i.e.\ each element is one less than the previous one. Using this terminology, we thus have that the sequence $L=(\ell_1,\ell_2,\ldots, \ell_n)$ is a concatenation of decrement runs, i.e.\ $L$ has the form $(x_1,x_1-1,x_1-2,\ldots, x_2,x_2-1,x_2-2,\ldots,\ldots,x_k,x_k-1,x_k-2,\ldots)$, with each $x_j=\ell_j$ for some head $j$. We can therefore store the matching statistics in compressed form as follows: 

\begin{definition}[Compressed matching statistics]
Let $R,S$ be two strings over $\Sigma$, and $\MS$ be the matching statistics of $S$ w.r.t.\ $R$. The {\em compressed matching statistics (\CMS) of $S$ w.r.t.\ $R$} is a data structure storing $(j,\MS[j])$ for each head $j$, and 
a predecessor data structure on the set of heads $H$. 
\end{definition}

We can use $\CMS$ to recover all values of $\MS$: 

\begin{lemma}\label{lemma:CMS}
Let $1\leq i\leq |S|$. Then $\MS[i] = (p_j+k,\ell_j-k)$, where $j = \pred_H(i)$ and $k = i-j$. 
\end{lemma}

\begin{proof}
Let $\ell_i$ be the length of the matching factor of $i$. Since there is a matching factor of length $\ell_j$ starting in position $j$ in $S$, this implies that $\ell_i \geq \max(0,\ell_j-k)$. If $\ell_i$ was strictly greater than $\ell_j-k$, this would imply the presence of another head between $j$ and $i$, in contradiction to $j = \pred_H(i)$. Since an occurrence of the matching factor $U_j$ of $j$ starts in position $p_j$ of $R$, therefore the matching factor $U'=U[k+1..\ell_j]$ of $i$ has an occurrence at position $p_j+k$. 
\end{proof}

\begin{figure}[h!]
    \centering
    \begin{tabular}{|r|r|r|r|r|r|r|r|r|r|r|r|r|r|r|r|r|r|r}
        \hline
        $i$ & 1 & 2 & 3 & 4 & 5 & 6 & 7 & 8 & 9 & 10 & 11 & 12 & 13 & 14 & 15 & 16 & 17 \\   
        \hline
        $R$ & {\tt T} & {\tt G} & {\tt A} & {\tt T} & {\tt G} & {\tt G} & {\tt C} & {\tt A} & {\tt C} & {\tt A} & {\tt G} & {\tt A} & {\tt T} & {\tt A} & {\tt C} & {\tt T} & {\tt \#}  \\
        \hline
        \hline
        $S$ & {\tt G} & {\tt A} & {\tt T} & {\tt G} & {\tt G} & {\tt C} & {\tt A} & {\tt C} & {\tt A} & {\tt T} & {\tt T} & {\tt G} & {\tt A} & {\tt T} & {\tt G} & {\tt G} & {\tt \$} \\
        $p_i$ & 2 & 3 & 4 & 5 & 6 & 7 & 8 & 9 & 12 & 13 & 1 & 2 & 3 & 4 & 5 & 6 & -1\\
        $\ell_i$ & 9 & 8 & 7 & 6 & 5 & 4 & 3 & 2 & 2 & 1 & 6 & 5 & 4 & 3 & 2 & 1 & 0 \\ 
        head & \checkmark & & & & & & & & \checkmark & & \checkmark & & & & & & \\
        \hline
        $q_i$ & 2 & 3 & 4 & 5 & 6 & 7 & 8 & 9 & 3 & 4 & 1 & 2 & 3 & 4 & 5 & 11 & 17\\
        i-head & \checkmark & & & & & & & & \checkmark & & \checkmark & & & & & \checkmark & \checkmark \\
        \hline
    \end{tabular}
    \caption{Example for the matching statistics and the data for the \CMS\ and the \eCMS. In the first two rows, we give $\MS$ of $S$ w.r.t.\ $R$, where $\MS[i]=(p_i,\ell_i)$. In row 3, we mark the heads (for the \CMS). In rows 4, we give the position $q_i$, defined by $ip(i)$, i.e.\ $q_i=\SA_R[ip(i)]$, where $ip(i)$ is the insert-point of suffix $S[i..]$ in the suffix array of $R$. In row $5$, we mark the insert-heads (for the \eCMS).}
    \label{fig:example_eCMS}
\end{figure}

\begin{figure}[ht]
    \centering

    \begin{tabular}{|r|r|r|l|}
        \hline
        & $i$ & $\SA_R$ & $R[{\SA_R[i]}..]$\\
        \hline
        & 1 & 17 & {\tt \#} \\ 
        & 2 & 8 & {\tt ACAGATACT\#} \\
        & 3 & 14 & {\tt ACT\#} \\
        & 4 & 10 & {\tt AGATACT\#} \\ 
        & 5 & 12 & {\tt \textcolor{blue}{AT}ACT\#} \\
      $\rightarrow$  & 6 & 3 & {\tt \textcolor{blue}{AT}GGCACAGATACT\#} \\ 
        & 7 & 7 & {\tt CACAGATACT\#} \\ 
        & 8 & 9 & {\tt CAGATACT\#} \\
        & 9 & 15 & {\tt CT\#} \\
        & 10 & 11 & {\tt GATACT\#} \\
        & 11 & 2 & {\tt GATGGCACAGATACT\#} \\ 
        & 12 & 6 & {\tt GCACAGATACT\#} \\ 
        & 13 & 5 & {\tt GGCACAGATACT\#} \\ 
        & 14 & 16 & {\tt T\#} \\
        & 15 & 13 & {\tt TACT\#} \\ 
        $\rightarrow$ & 16 & 1 & {\tt \textcolor{blue}{TGATGG}CACAGATACT\#} \\
        & 17 & 4 & {\tt TGGCACAGATACT\#} \\
        \hline
    \end{tabular}
    
    \caption{Details of computation of the matching statistics from Figure \ref{fig:example_eCMS}. We highlight in blue the matching factors for the indices $i=9$ (matching factor ${\tt AT}$, mismatch character ${\tt T}$) and $11$ (matching factor ${\tt TGATGG}$, mismatch character ${\tt \$}$). The arrows represent the insert-points.} 
    \label{fig:insertion-point}
\end{figure}

\begin{example}\label{ex:1} Consider the reference $R={\tt TGATGGCACAGATACT}$ and $S=$ ${\tt GATGGCACATTGATGG}$. 
The \CMS\ of $S$ w.r.t.\ $R$ is: $(1,2,9), (9,12,2), (11,1,6)$, see Figure~\ref{fig:example_eCMS}. 
\end{example}

From Lemma~\ref{lemma:CMS} and the properties of the predecessor data structure on the set of heads we get: 

\begin{proposition}\label{prop:CMS}
Let $R,S$ be two strings over $\Sigma$. We can store the matching statistics of $S$ w.r.t.\ $R$ in $\Oh(\chi)$ space such that any entry $\MS[i]$, for $1\leq i \leq |S|$, can be accessed in $\Oh(\log\log \chi)$ time, where $\chi=|H|$ is the number of heads. 
\end{proposition}

For some statistics on the number $\chi$ of heads, see the end of Sec.~\ref{sec:enhanced_cms}. 

%%%%%%%%%%%%%%%%%%%%%%%%%%%%%%%%%%%%%%%%%%%%%%%%%%%%%%%%%%%%%%%%%%%%%%%%%%%%%%%%%%%%%%%
%% SECTION 
%%%%%%%%%%%%%%%%%%%%%%%%%%%%%%%%%%%%%%%%%%%%%%%%%%%%%%%%%%%%%%%%%%%%%%%%%%%%%%%%%%%%%%%

\subsection{Enhancing the CMS}\label{sec:enhanced_cms}

Let $R,S$ be two strings over $\Sigma$, and $\MS$ the matching statistics of $S$ w.r.t.\ $R$. We now assume that all characters that occur in $S$ also occur in $R$ (see Sec.~\ref{sec:implementation}). Let $\SA_R$ be the suffix array of $R$. For position $i$ of $S$, let $U\neq \epsilon$ be the matching factor and $c$ the mismatch character of $i$. We want to compute the position that the suffix $S[i..]$ would have in $\SA_R$ if it was present.  To this end, we define the {\em insert point} of $i$, $ip(i)$,  as follows:  
\begin{align*} 
ip(i) = 
\begin{cases} 1 & \text{ if $U=\epsilon$},\\  
\max\{ j \mid U \text{ occurs in } \SA_R[j] \text{ and } R[\SA_R[j]..] < Uc\} & \text{ if this set is non-empty,}\\
 \min\{j \mid U \text{ occurs in } \SA_R[j]\} & \text{ otherwise.} 
 \end{cases} 
\end{align*}

In other words, the insert point is the lexicographic rank, among all suffixes of $R$, of the next smaller occurrence of $U$ in $R$ if such an occurrence exists, and of the smallest occurrence of $U$ in $R$ otherwise. Note that case 1 (where $U=\epsilon$) only happens for end-of-string characters. The insert point is well-defined for every $i$ because $\#$ is smaller than all other characters, including other end-of-string characters. Observe that the insert point of $i$ always lies within the $U$-interval of $\SA_R$. For an example, see Fig.~\ref{fig:insertion-point}. 

We will later use the insert points to bucket suffixes. First we need to slightly change the definition of our compressed matching statistics. We will add more information to the heads: we add the mismatch character and replace the position entry $p_i$, which gives just some occurrence of the matching factor, by the specific occurrence $q_i$ given by the insert point. This will imply adding more heads, so our data structure may increase in size. 

To this end, we define $j$ to be an {\em insert-head} if $\SA_R[ip(j)] \neq \SA_R[ip(j-1)]+1$.  Note that, in particular, all heads are also insert-heads, but it is possible to have insert-heads $j$ which are not heads, namely where $\ell_j = \ell_{j-1}-1$. 

\begin{definition}[Enhanced compressed matching statistics]
Let $R,S$ be two strings over $\Sigma$. Define the {\em enhanced matching statistics} as follows: for each $1\leq i \leq |S|$, let $\ems(i) = (q_i,\ell_i,x_i,c_i)$, where $q_i = \SA_R[ip(i)]$, $\ell_i$ is the length of the matching factor $U$ of $i$, $c_i$ is the mismatch character, and $x_i\in \{S,L\}$ indicates whether $Uc_i$ is smaller (S) or greater (L) than $R[q_i..]$. 
The {\em enhanced compressed matching statistics (\eCMS) of $S$ w.r.t.\ $R$} is a data structure storing $(j,\ems(j))$  for each insert-head $j$, and a predecessor data structure on the set of insert-heads $H'$. 
\end{definition}

\begin{example}\label{ex:2} Continuing with Example~\ref{ex:1}, the enhanced \CMS\ 
of $S$ w.r.t.\ $R$ is: 
$(1,2,9,L,{\tt T})$, $(9,3,2,L,{\tt T})$, $(11,1,6,S,{\tt \$}),$ $(16,11,1,S,{\tt \$})$, $(17,17,0,L,{\tt \$})$, see Figure~\ref{fig:example_eCMS}. 
\end{example}

We will need some properties of the insert point in the following: 

\begin{observation}\label{obs:ip}
Let $ip(i)$ be the insert point of $i$, and $\ems(i) = (q_i,\ell_i,x_i,c_i)$. 

\begin{enumerate}
    \item $ip(i) = ip(i')$ if and only if $q_i = q_{i'}$, 
    \item if $x_i=S$ then $R[\SA_R[ip(i)-1]..] < S[i..] < R[\SA_R[ip(i)]..] = R[q_i..]$, 
    \item if $x_i=L$ then $R[q_i..] = R[\SA_R[ip(i)]..] < S[i..] < R[\SA_R[ip(i)+1]..]$. 
\end{enumerate}

\end{observation}

The enhanced CMS can be used in a similar way as the CMS to recover the enhanced matching statistics (including the matching statistics) of each $i$. Denote by $\ihead(i)$ the next insert-head to the left of $i$, i.e.\ $\ihead(i) = \max\{j\leq i \mid j \text{ is an insert-head}\}$. Note that $\ihead(i) = \pred_{H'}(i)$. 

\begin{lemma}\label{lemma:eCMS}
Let $1\leq i \leq |S|$, let $\eCMS$ be the enhanced CMS of $S$ w.r.t.\ $R$. Let $j = \ihead(i)$, $k=i-j$, and $\ems(j) = (q_j,\ell_j,x_j,c_j)$.  Then $\ems(i) = (q_j+k,\ell_j-k,x_j,c_j)$, and $ip(i) = \ISA_R[q_j+k]$. In particular, $q_j+k$ is an occurrence and $\ell_j-k$ is the length of the matching factor of $i$ (in other words, the matching statistics entry $\MS[i]$). 
\end{lemma}

\begin{proof}
Analogous to Lemma~\ref{lemma:CMS}, resp.\ straightforward from the definitions. 
\end{proof}

Similarly to the \CMS\ (cp. Prop.~\ref{prop:CMS}), the enhanced \CMS\ allows access to all values for every index $i$, using space $\Oh(\chi')$ and time $\Oh(\log \log \chi')$, where $\chi'=|H'|$ is the number of insert-heads. Again, this is due to the fact that the predecessor data structure on the set $H'$ of insert-heads allows retrieving $\pred_{H'}(i) =\ihead(i)$ in $\Oh(\log \log |H'|)$ time, and the values of $\ems(i)$ can then be computed in $\Oh(1)$ time (Lemma~\ref{lemma:eCMS}). 

\medskip

We close this subsection by remarking that for a collection of similar genomes, one can expect the number of heads to be small. Indeed, on a 500MB viral genome data set (see Section~\ref{sec:experiments}) containing approximately 10,000 SARS-cov2 genomes, we observed the number of heads to be 5,326,226 (100x less than the input size) and the number of insert heads to be 6,537,294.

\subsection{Computing the CMS}\label{sec:computing_cms}

It is well known that the matching statistics of $S$ w.r.t.\ $R$ can be computed in time $O(|R| + |S|\log\sigma)$ and $O(|R|)$ space by using, for example, the suffix tree of $R$, as described in Chang and Lawler's original paper~\cite{CL94}. Since then, several authors have described similar algorithms for computing matching statistics, all focussed on reducing space requirements via the use of compressed indexes instead of the suffix tree~\cite{AKO04,OGK10,BCD18}. These algorithms all incur the slowdowns typical of compressed data structures.

In our setting, where end-to-end runtime is the priority, it is the speed at which the matching statistics can be computed (rather than working space) that is paramount. Moreover, because the size of the reference is generally small relative to the total length of all the strings $S_i \in {\cal C}$, we have some freedom to use large index data structures on $R$ to compute the matching statistics, without overall memory usage getting out of hand. With these factors in mind, we take the following approach to computing CMS. The algorithm is similar to that of Chang and Lawler, but makes use of array-based data structures rather than the suffix tree. 

Recall that, given the suffix array $\SA_R$ of string $R$ and a substring $\Y$ of $R$, the $\Y$-interval  is the interval $\SA_R[s..e]$ that contains all suffixes having $\Y$ as a prefix. 

\begin{definition}[Right extension and left contraction]
For a character $c$ and a string $\Y$, the computation
of the $\Yc$-interval from the $\Y$-interval is called a {\em right extension} and the computation of
the $\Y$-interval from $\cY$-interval is called a {\em left contraction}. 
\end{definition}

We remark that a left contraction is equivalent to following a (possibly implicit) suffix link in the suffix tree of $R$ and a right extension is a downward movement (either to a child or along an edge) in the suffix tree of $R$.

Given a $Y$-interval, because of the lexicographical ordering on the $\SA_R$, we can implement a right extension to a $\Yc$-interval in $O(\log |R|)$ time by using a pair of binary searches (with $c$ as the search key), one to find the lefthand end of the $\Yc$-interval and another to find the righthand end. If a right extension is empty then there are no occurrences of $\Yc$ in $R$, but we can have the binary search return to us the insert point where it would have been in $\SA_R$. 

On the other hand, given a $cY$-interval, $\SA_R[s..e]$, we can compute the $Y$-interval (i.e. perform a left contraction) in the following way. Let the target $Y$-interval be $\SA_R[x..y]$. Observe that both $\SA_R[s]+1$ and $\SA_R[e]+1$ must be inside the $Y$-interval, $SA_R[x..y]$---that is, $s' = \ISA_R[\SA_R[s]+1] \in [x..y]$ and $e' = \ISA_R[\SA_R[e]+1] \in [x..y]$. To finish computing $\SA_R[x..y]$ from $\SA_R[s'..e']$ there are two cases to consider. Firstly, if $s' = e'$ and $|Y| > \LCP_R[s']$, then $\SA_R[s']$ is the only occurrence of $Y$ and we are done (the $Y$-interval is a singleton). Alternatively, $s' \ne e'$ and we compute $\SA_R[x..y]$ using $\NSV/\PSV$ queries on $LCP_R$, in particular $\SA_R[x..y] = \SA_R[\PSV(\LCP_R,s')..\NSV(\LCP_R,e')]$.

With these ideas in place, we are ready to describe the matching statistics algorithm.
We first compute $\SA_R$, $\ISA_R$, and $\LCP_R$ for $R$ and preprocess $\LCP_R$ for $\NSV/\PSV$ queries. The elements of the $\MS$ will be computed in left-to-right order, $\MS[1], \MS[2], \ldots, \MS[|S|]$. Note that this makes it trivial to save only the heads (or iheads) and so compute the CMS (or eCMS) instead.
To find $\MS[1]$ use successive right extensions   starting with the interval $SA_R[1..|R|]$, searching with successive characters of $S[1..]$ until the right extension is empty, at which point we know $\ell_1$ and $p_1$. 
At a generic step in the algorithm, immediately after computing $\MS[i]$, we know the interval $\SA_R[s_i..e_i]$ containing all the occurrences of $R[p_i..p_i+\ell_i-1]$. To compute $\MS[i+1]$ we first compute the left contraction of $\SA_R[s_i..e_i]$, followed by as many right contractions as possible until $\ell_{i+1}$ and $p_{i+1}$ are known.

\medskip

When profiling an implementation of the above algorithm, we noticed that very often the sequence of right extensions ended with a singleton interval (i.e., an interval of size one) and so was the interval reached by the left contraction that followed. In terms of the suffix tree, this corresponds to the match between $R$ and the current suffix of $S_i$ being inside a leaf branch. This frequently happens on genome collections because each sequence is likely to have much longer matches with other sequences (in this case with $R$) than it does with itself (a single genome tends to look fairly random, at least by string complexity measures).

A simple heuristic to exploit this phenomenon is to compare $\ell_i$ to the maximum value in the entire $\LCP_R$ array of $R$ immediately after $\MS[i]$ has been computed. If $\ell_i-1 > \max(\LCP_R)$ then $\ISA_R[p_i+1]$ will also be inside a leaf branch (i.e., the left contraction will also be a singleton interval), and so the left contraction can be computed trivially as $\ISA_R[p_i+1]$---with no subsequent $\NSV/\PSV$ queries or access to $\LCP_R$ required to expand the interval. Although this gives no asymptotic improvement, there is potential gain from the probable cache miss(es) avoided by not making random accesses to those large data structures.

On a viral genome data set (see Section~\ref{sec:experiments}), $\max(\LCP_R)$ was 14, compared to an average $\ell_i$ value of over $1,100$, and this heuristic saved lots of computation. On a human chromosome data set, however, $\max(\LCP_R)$ was in the hundreds of thousands, and so we generalized the trick in the following way. We divide the LCP array up into blocks of size $b$ and compute the minimum of each block. These minima are stored in an array $M$ of size $|R|/b$, and $b$ is chosen so that $M$ is small enough to comfortably fit in cache. Now, when transitioning from $\MS[i]$ to $\MS[i+1]$, if $\ell_i > M[\ISA_R[p_i+1]/b]$ then there is a single match corresponding to $\MS[i+1]$, which we compute with right extensions. 
This generalized form of the heuristic has a consistent and noticeable effect in practice. For a 500MB viral genome data set its use reduced CMS computation from 12.23 seconds to 2.34 seconds. On the human chromosome data set the effect is even more dramatic: from 76.50 seconds down to 7.14 seconds.

%%%%%%%%%%%%%%%%%%%%%%%%%%%%%%%%%%%%%%%%%%%%%%%%%%%%%%%%%%%%%%%%%%%%%%%%%%%%%%%%%%%%%%%
%% SECTION 
%%%%%%%%%%%%%%%%%%%%%%%%%%%%%%%%%%%%%%%%%%%%%%%%%%%%%%%%%%%%%%%%%%%%%%%%%%%%%%%%%%%%%%%

\section{Comparing two suffixes via the enhanced CMS}\label{sec:comparing}

We will now show how to use the enhanced CMS of the collection ${\cal C}$ w.r.t.\ $R$ to define a partial order on the set of suffixes of strings in ${\cal C}$ (Prop.~\ref{prop:PO}), and how to break ties when the entries are identical (Lemma~\ref{lemma:sufcomp_3}). These results can then be used either directly to determine the relative order of any two of the suffixes (Prop.~\ref{prop:suffix_comparison}), or as a way of inducing the complete order once that of the subset of the insert-heads has been determined (Prop.~\ref{prop:suffix_sort}).

We will prove Prop.~\ref{prop:PO} via two lemmas. Recall that in the \eCMS\ we only have the entries referring to the insert-heads; however, Lemma~\ref{lemma:eCMS} tells us how to compute them for any position. 

\begin{lemma}\label{lemma:sufcomp_1}
Let $1\leq d,d' \leq m$ and $1\leq i\leq |S_d|$, $1\leq i'\leq |S_{d'}|$. If $ip(d,i) < ip(d',i')$, then $S_d[i..] < S_{d'}[i'..]$.
\end{lemma}

\begin{proof}
If $ip(d',i')-ip(d,i)>1$, then there exists an index $j$ s.t.\ $ip(d,i)<j<ip(d',i')$, and therefore $S_d[i..] < R[\SA_R[ip(d,i)+1]..] \leq  R[\SA_R[j]..] \leq R[\SA_R[ip(d',i')-1]..] < S_{d'}[i'..]$. Now let $ip(d',i') = ip(d,i)+1$.  If $x_{d,i}=S$, then $S_d[i..]<R[\SA_R[ip(d,i)]..] = R[\SA_R[ip(d',i')-1]..] < S_{d'}[i'..]$, by Obs.~\ref{obs:ip}. Similarly, if $x_{d',i'}=L$, then $S_d[i..]<R[\SA_R[ip(d,i)+1]..] = R[\SA_R[ip(d',i')]..] < S_{d'}[i'..]$. Finally, let $x_{d,i}=L$ and $x_{d',i'}=S$. Then $R[\SA_R[ip(d,i)]..] < S_d[i..],S_{d'}[i'..] < R[\SA_R[ip(d,i)+1]..] = R[\SA_R[ip(d',i')]..]$. Let $U$ be the matching factor of $(d,i)$, $U'$ that of $(d',i')$, and $V = \lcp(U,U')$, the longest common prefix of the two. $V$ cannot be equal to $U'$ because then $U'$ would be a proper prefix of $U$, but $ip(d',i')$ is the smallest occurrence in $R$ of $U'$. If $V=U$, then $U$ is a proper prefix of $U'$, and by definition of $ip(d',i')$, the character following $U$ in $U'$ is strictly greater than the mismatch character $c_i$ of $(d,i)$. Finally, if $V$ is a proper prefix both of $U$ and of $U'$, then the character following $V$ in $U$ is smaller than the one following $V$ in $U'$, therefore $U<U'$. Since $U$ is a prefix of $S_d[i..]$ and $U'$ is a prefix of $S_{d'}[i'..]$, and neither is prefix of the other, this implies $S_d[i..] < S_{d'}[i'..]$. 
\end{proof}

\begin{lemma}\label{lemma:sufcomp_2}
Let $1\leq d,d' \leq m$ and $1\leq i\leq |S_d|$, $1\leq i'\leq |S_{d'}|$, and $ip(d,i) = ip(d',i')$. 
\begin{enumerate}
    \item If $\ell_{d,i} < \ell_{d',i'}$ and $x_{d,i}=S$, then $S_d[i..] < S_{d'}[i'..]$. 
    \item If $\ell_{d,i} < \ell_{d',i'}$ and $x_{d,i}=L$, then $S_{d'}[i'..] < S_d[i..]$. 
    \item If $\ell_{d,i} = \ell_{d',i'}$ and $x_{d,i}=S$ and $x_{d',i'}=L$, then $S_d[i..] < S_{d'}[i'..]$. 
    \item If $\ell_{d,i} = \ell_{d',i'}$ and $x_{d,i}=x_{d',i'}$ and $c_{d,i}<c_{d',i'}$, then $S_d[i..] < S_{d'}[i'..]$. 
\end{enumerate}
\end{lemma}

\begin{proof}
{\em 1.,2.:} Let $U$ be the matching factor of $i$, and $U'$ that of $i'$. Since $\ell_{d,i}<\ell_{d',i'}$, this implies that $U$ is a proper prefix of $U'$. If $x_{d,i}=S$, then the mismatch character $c_{d,i}$ is smaller than the character following $U$ in $U'$, therefore $S_d[i..] < S_{d'}[i'..]$. If $x_{d,i}=L$, then it is greater, and thus $S_{d'}[i'..] < S_d[i..]$. {\em 3.} follows directly from Observation~\ref{obs:ip}, since now $S[i..] < R[\SA_R[ip(i)]..] < S[i'..]$. {\em 4.: } Now both suffixes start with the same matching factor $U$, followed by different mismatch characters, which define their relative order.
\end{proof}

These two lemmas in fact imply the following: 

\begin{proposition}\label{prop:PO}
The conditions of Lemmas~\ref{lemma:sufcomp_1} and~\ref{lemma:sufcomp_2} result in a partial order of the suffixes of strings in ${\cal C}$, of which the lexicographic order is a refinement. 
\end{proposition}

What happens if two suffixes $S_d[i..]$ and $S_{d'}[i'..]$ have the same values of the enhanced matching statistics, i.e.\ $\ems(d,i) = \ems(d',i')$? The next lemma says that in this case, the relative order of the two suffixes is decided by the relative order of the heads preceding their respective mismatch characters. 

\begin{lemma}\label{lemma:sufcomp_3}
Let $1\leq d,d' \leq m$ and $1\leq i\leq |S_d|$, $1\leq i'\leq |S_{d'}|$. If $ip(d,i) = ip(d',i')$, $\ell_{d,i} = \ell_{d',i'}$, $x_{d,i}=x_{d',i'}$, and $c_{d,i}=c_{d',i'}$,  then $S_d[i..] < S_{d'}[i'..]$ if and only if $S_d[j..] < S_{d'}[j'..]$, where $(d,j) = \ihead(d,i+\ell_i)$ and $(d',j') = \ihead(d',i'+\ell_{i'})$. 
\end{lemma}

\begin{proof}
We will prove that the relative position of the insert-head of $i$'s and $i'$'s mismatch character is the same, i.e.\ that $j-i = j'-i'$. The claim then follows. 

First note that $j>i$. This is because the matching factor of position $i$ ends in position $i+\ell_{d,i}-1$, so there must be a new insert-head after $i$ and at most at $i+\ell_{d,i}$, the position of the mismatch character. Similarly, $j'>i'$. The fact that $j=\ihead(i+\ell_{d,i})$ implies that there is a matching factor starting in position $j$ which spans the mismatch character $c=c_{d,i}=c_{d',i'}$. Let's write $Vc$ for the prefix of length $i+\ell_{d,i}-j$ of this matching factor. $V$ is a suffix of the matching factor $U$ of position $i$, but $Vc$ is not. However, $Vc$ is also a prefix of $S_{d'}[i'..]$. Therefore, $j'=i'+(j-i)$ is also an insert-head in $S_{d'}$. An analogous argument shows that any insert-head between $i'$ and $i'+\ell_{d',i'}$ in $S_{d'}$ is also an insert-head in $S_d$, in the same relative position. 
\end{proof}

\begin{proposition}\label{prop:suffix_comparison}
Let $R,S_1,\ldots,S_m$ be strings over $\Sigma$. Using the enhanced \CMS\ of ${\cal C}=\{S_1,\ldots,S_m\}$ w.r.t.\ $R$, we can decide, for any $1\leq d,d' \leq m$ and $1\leq i\leq |S_d|$, $1\leq i'\leq |S_{d'}|$, the relative order of $S_d[i..]$ and $S_{d'}[i'..]$ in $\Oh(\log\log\chi'\cdot \max_d \{\text{no.\ of insert-heads of } S_d\})$ time. 
\end{proposition}

\begin{proof}
Let $(d,j) = \ihead(d,i+\ell_i)$ and $(d',j') = \ihead(d',i'+\ell_{i'})$. From Lemma~\ref{lemma:eCMS} we get the four \eCMS-entries of $(d,i)$ and $(d',i')$, namely the insert positions $q_i$ resp.\ $q_{i'}$, the length of the matching factor, whether the mismatch characters is smaller or larger, and the mismatch character itself. If any of these differ for the two suffixes, then Lemmas~\ref{lemma:sufcomp_1} and~\ref{lemma:sufcomp_2} tell us their relative order. This check takes $\Oh(1)$ time. Otherwise, Lemma~\ref{lemma:sufcomp_3} shows that the relative order is determined by the next relevant heads. Iteratively applying the three lemmas, in the worst case, takes us through all heads for the strings $S_d$ and $S_{d'}$. 
\end{proof}

Instead of using Prop.~\ref{prop:suffix_comparison}, we will use these lemmas in the following way. We will first sort only  the insert-heads. The following proposition states that this suffices to determine the order of any two suffixes in constant time. 

\begin{proposition}\label{prop:suffix_sort}
Given the insert-heads in sorted order, the relative order of any two suffixes can be determined in $\Oh(\log\log\chi')$ time, where $\chi'$ is the number of insert-heads. 
\end{proposition}

\begin{proof}
Follows from Lemmas~\ref{lemma:sufcomp_1},~\ref{lemma:sufcomp_2}, and~\ref{lemma:sufcomp_3}, since all checks take constant time, and each of the two predecessor queries take $\Oh(\log\log\chi')$ time. 
\end{proof}

%%%%%%%%%%%%%%%%%%%%%%%%%%%%%%%%%%%%%%%%%%%%%%%%%%%%%%%%%%%%%%%%%%%%%%%%%%%%%%%%%%%%%%%
%% SECTION 
%%%%%%%%%%%%%%%%%%%%%%%%%%%%%%%%%%%%%%%%%%%%%%%%%%%%%%%%%%%%%%%%%%%%%%%%%%%%%%%%%%%%%%%

\section{Putting it all together}\label{sec:putting}

A high-level view of our algorithm is as follows. We first partially sort the insert-heads, then use this partial sort to generate a new string, whose suffixes we sort with an existing suffix sorting algorithm. This gives us a full sort of the insert heads. We then use this to sort the $S^*$-suffixes of the collection. Finally, we induce the remaining suffixes of the collection using the $S^*$-suffixes. 
We next give a schematic description of the algorithm. 

\medskip

\begin{quote}
\hrule
{\tt Algorithm 1}\\
{\bf input:} string collection ${\cal C}$, reference string $R$\\
{\bf output:} the $\GSA$ of ${\cal C}$
\hrule
\begin{itemize}
    \item {\bf Phase 1 - Augmenting and constructing data structures on $R$:} Preprocess $R$ (``augmenting'', see Sec.~\ref{sec:implementation}). Compute the data structures $\SA_R, \ISA_R, \PLCP_R, \LCP_R$ and the  RMQ-data structure for \PSV- and \NSV-queries on $\LCP_R$. 

    \item {\bf Phase 2 - Computing the \eCMS:} Compute the \eCMS\ of ${\cal C}$, as described in Sec.~\ref{sec:computing_cms}. 
    
    \item {\bf Phase 3 - Bucketing:} Identify the $S^*$-suffixes in ${\cal C}$ via a backward linear scan of ${\cal C}$. Bucket $S^*$-suffixes $i$ according to $ip(i)$, computed using the \eCMS\ (Lemma~\ref{lemma:eCMS}). 
    
    \item {\bf Phase 4 - Sorting the insert-heads:} 
    \begin{itemize}
        \item bucket the insert-heads according to their insert point;
        \item for each bucket $B$, partially sort $B$, according to Lemmas~\ref{lemma:sufcomp_1} and~\ref{lemma:sufcomp_2};
        \item rename insert-heads according to lexicographic rank of substring stretching up to the mismatch character (metacharacters are $S_d[j..j+\ell_{d,j}]$);
        \item generate new string $C$ as concatenation of these metacharacters; 
        \item compute the suffix array of $C$, map back to corresponding suffixes of ${\cal C}$.  
    \end{itemize}

    \item {\bf Phase 5 - Fully sorting the $S^*$-suffixes:} 
    for each bucket $B$ from Phase 3, sort $B$, according to Lemmas~\ref{lemma:sufcomp_2} and~\ref{lemma:sufcomp_3}

    \item {\bf Phase 6 - Inducing the \GSA:} With two scans, induce $L$-suffixes, induce $S$-suffixes. 

\end{itemize}
\hrule
\end{quote}

\medskip

We next give a worst-case asymptotic analysis of the algorithm. 

\begin{proposition}
Algorithm 1 computes the \GSA\ of a string collection ${\cal C}$ of total length $N$ in worst-case time $\Oh(N\log N)$. 
\end{proposition}

\begin{proof}
Let $|R|=n$. Phase 1 takes $\Oh(n+N)$ time, since constructing all data structures on $R$ can be done in linear time in $n$ and scanning the collection ${\cal C}$ takes time $\Oh(N)$. Phase 2 takes time $\Oh(N\log n)$ using the algorithm from Sec.~\ref{sec:computing_cms}. In Phase 3, identifying the $S^*$ suffixes, takes time $\Oh(N)$. Since at this point, the $\eCMS$ is in text-order, identifying $\ihead(i)$ takes constant time, also computing the insert-point takes constant time, so altogether $\Oh(N)$ time. In Phase 4, all steps are linear in $\chi'$, the number of insert-heads, including the partial sort of the buckets, since this can be done with radix-sort (three passes over each bucket), so this phase takes time $\Oh(\chi')$. Phase 5 takes time $\Oh(|B|\log|B|)$ for each bucket $B$, thus $\Oh(N\log |B_{\max}|)$ for the entire collection, where $B_{\max}$ is a largest bucket. Since all strings in the collection are assumed to be highly similar to the reference, the size of the buckets can be expected to vary around the number of strings in the collection $m$; however, in the worst case the largest bucket can be $\Theta(N)$. Finally, Phase 6 takes linear time $\Oh(N)$. Altogether, the running time is dominated by Phase 5, $\Oh(N\log N)$. 
\end{proof}

%%%%%%%%%%%%%%%%%%%%%%%%%%%%%%%%%%%%%%%%%%%%%%%%%%%%%%%%%%%%%%%%%%%%%%%%%%%%%%%%%%%%%%%
%% SECTION 
%%%%%%%%%%%%%%%%%%%%%%%%%%%%%%%%%%%%%%%%%%%%%%%%%%%%%%%%%%%%%%%%%%%%%%%%%%%%%%%%%%%%%%%

\section{Implementation details}\label{sec:implementation}

In Phase 1, the augmentation step involves, for every character $c$ not occurring in $R$ but occurring in ${\cal C}$, appending $c^{n_c}$ to $R$, where $n_c$ is the length of the longest run of $c$ in ${\cal C}$. This avoids having $0$-length entries in the matching statistics and is necessary in order to have a well defined $ip$.

To compute $\SA_R$ in Phase 1, we use {\tt sais}~\cite{sais-lite} as implemented by Yuta Mori, a well engineered version of SAIS \cite{NZC11}, which was chosen due to its consistent speed on many different inputs.
For the computation of $\PLCP_R$ and $\LCP_R$ we use the $\Phi$ method~\cite{KMP09}. This is the fastest method to compute the LCP array we know of.
We  constructed  the  data structure  of  C{\'a}novas  and  Navarro~\cite{CN10} for NSV/PSV queries on the LCP array, as it has low space overheads and was fast to query and initialize.

For the predecessor data structure, we use the following two-layered approach in practice (rather than~\cite{W83}). We sample every $b$th head starting position and store these in an array. In a separate array we store a differential encoding of all head positions. The array of differentially encoded starting positions takes
32 bits per entry. Predecessor search for a position $x$ proceeds by first binary searching in the sampled array to
find the predecessor sample at index $i$ of that array. We then access the differentially encoded array starting at
index $ib$ and scan, summing values until the cumulative sum is greater than $x$, at which point we know the predecessor. This takes $O(\chi'/b + b)$ time, where $\chi'$ is the number of insert-heads.

For Phase 4, when we have to sort $C$ (the concatenation of metacharacters representing partially sorted heads), we use a SACA-K implementation that handles integer alphabets~\cite{LouzaGT17}. This choice was made because of this algorithm's low space requirement, in particular, $\Oh(K)$, where $K$ is the number of distinct $\ems$-entries of insert-heads in $H'$ (note $K = O(\chi')$).

%%%%%%%%%%%%%%%%%%%%%%%%%%%%%%%%%%%%%%%%%%%%%%%%%%%%%%%%%%%%%%%%%%%%%%%%%%%%%%%%%%%%%%%
%% SECTION 
%%%%%%%%%%%%%%%%%%%%%%%%%%%%%%%%%%%%%%%%%%%%%%%%%%%%%%%%%%%%%%%%%%%%%%%%%%%%%%%%%%%%%%%

\section{Experiments}\label{sec:experiments}

We implemented our algorithm for computing the generalized suffix array in C++. Our prototype implementation, \sacamats, is available at \url{https://github.com/fmasillo/sacamats}. The experiments were conducted on a laptop equipped with 16GB of RAM DDR4-2400MHz and an Intel(R) Core(R) i5-8250U@3.4GHz with 6MB of cache. The operating system is Ubuntu 20.04 LTS, the compiler used is {\tt g++} version 9.4.0 with options {\tt -std=c++17 -O3 -funroll-loops} enabled. 

In the following experiments, we compare $\sacamats$ to two well known suffix array construction tools, both implementations by Yuta Mori~\cite{sais-lite, divsufsort}. The first, \sais, is an implemenation of the well-known SAIS algorithm by Nong et al.~\cite{NZC11}; the second, $\divsufsort$~\cite{0001K17}, is perhaps the most widely used tool for suffix array construction. We also compare against $\gsufsort$~\cite{LouzaTGPR20}, which is an extension of the SACA-K algorithm~\cite{Nong13} to a collection of strings, and to $\bigBWT$ \cite{BoucherGKLMM19}, a tool computing the BWT and the suffix array, designed specifically for highly repetitive data. 

\subsection{Datasets}
For our tests, we used two publicly available datasets, one consisting of copies of human chromosome 19 from the 1000 Genomes Project~\cite{1000genomes}, and another of copies of SARS-CoV2 genomes taken from NCBI Datasets\footnote{\url{https://www.ncbi.nlm.nih.gov/datasets/coronavirus/genomes/}}. For both datasets we selected subsets of different sizes in order to study the scalability of our algorithm. The sizes are 250MB, 500MB, 800MB and 1GB. More information can be found in Table~\ref{tab:datasets}.

We observe that on both datasets the number of i-heads is around 100x less than the input size, and on {\tt chr19} it is 8x less than the number of BWT runs.

\begin{table}[th]
    \centering
    \begin{tabular}{|l|l|r|r|r|r|}
         \hline
         Name & Description & $\sigma$ & $r$ & no.\ of $S^*$-suffixes & no.\ of i-heads \\
         \hline
         {\tt chr19} & Human Chromosome 19 & 5 & 32\,018\,267 & 129\,129\,636 & 4\,220\,033 \\
         {\tt sars-cov2} & SARS-CoV2 genome & 14 & 351\,596 & 143\,588\,463 & 6\,537\,294 \\
         \hline
    \end{tabular}
    \caption{Datasets used in experiments. In column 3, we specify the alphabet size $\sigma$, in column 4 the number $r$ of runs of the BWT, in column 5 the number of $S^*$-suffixes, and column 6 the number of insert-heads. In our experiments we use prefixes of each dataset up to 1GB. The last three columns refer to the 500MB prefix. }
    \label{tab:datasets}
\end{table}

\subsection{Results}

In Figures \ref{fig:chr19} and \ref{fig:covid}, information about running time for both datasets is displayed. The line plot represents a direct comparison of different algorithms, whereas the stacked bar plot is to visualize how much each phase of $\sacamats$ takes w.r.t. the total running time (cp.~Sec.~\ref{sec:putting}). 

These tools all produce slightly different outputs: $\sais$ and $\divsufsort$ output the $\SA$, $\gsufsort$ and $\sacamats$ the $\GSA$, and $\bigBWT$ both the $\BWT$ and the $\SA$. Because of these differences, if one were to write to disk each result, the running time would be affected accordingly by the size of the output. Therefore, we only compare the building time, i.e.\ the time spent constructing the SA and storing it in a single array in memory, without the time spent writing it to disk. For this reason, we made slight changes to the $\bigBWT$ code to enable storing the $\SA$ in main memory.  

By looking at the line plots, one can see that $\sacamats$ is competitive in both scenarios, i.e., it is faster than all tools on {\tt sars-cov2},  except $\bigBWT$. The same is true for {\tt chr19}, where it is the fastest method, especially on larger inputs, but here the main competitor becomes $\divsufsort$. More precisely, for the first dataset ({\tt chr19}) and considering 1GB of data, $\sacamats$ takes less than a third of the time of $\gsufsort$, is 20\% faster than $\sais$, 12\% faster than $\bigBWT$, and 5\% faster than $\divsufsort$. For the second dataset ({\tt covid}), $\sacamats$ takes again less than a third of the time of $\gsufsort$, is 37\% faster than $\divsufsort$, 16\% faster than $\sais$, and 30\% slower than $\bigBWT$.

Shifting our attention to the stacked bar plots, Figure~\ref{fig:chr19-phases} indicates that a lot of time is spent in the first phase, consisting in the augmentation of $R$ and the construction of various data structures for the augmented version of $R$. In the setting of DNA strings it is not too hard to think that the augmentation process will not elongate $R$, due to the very restricted alphabet. If the application lends itself to it, one could compute beforehand all the data structures listed in Phase 1, gaining roughly 20 seconds of run time. In our experiment on {\tt chr19} we would then be clearly the best algorithm, further distancing from the others. Alternatively, the common method of replacing {\tt N} symbols with random nucleotide symbols would be another way to speed up this phase.

Finally, we comment on memory usage, which is highest for $\sacamats$ and $\gsufsort$ at 8 bytes per input symbol, and 4 bytes per input symbol for $\divsufsort$ and $\sais$, and $\bigBWT$ (including the 4 bytes per input symbol of the $\SA$ when it is saved in memory, see above). We have not yet optimized for memory usage and note that a semi-external implementation of our approach, in which buckets reside on disk, presents itself as an effective way to reduce main memory usage. In all phases, the actual working set---the amount of data active in main memory---is small (for the most part, proportional to the number of i-heads), and other authors have shown that the inducing phase is amenable to external memory, too~\cite{KKPZ17}. We leave these optimizations as future work.

\begin{figure}[hb]
\centering
\begin{subfigure}{.5\textwidth}
    \centering
	\includegraphics[width=\textwidth]{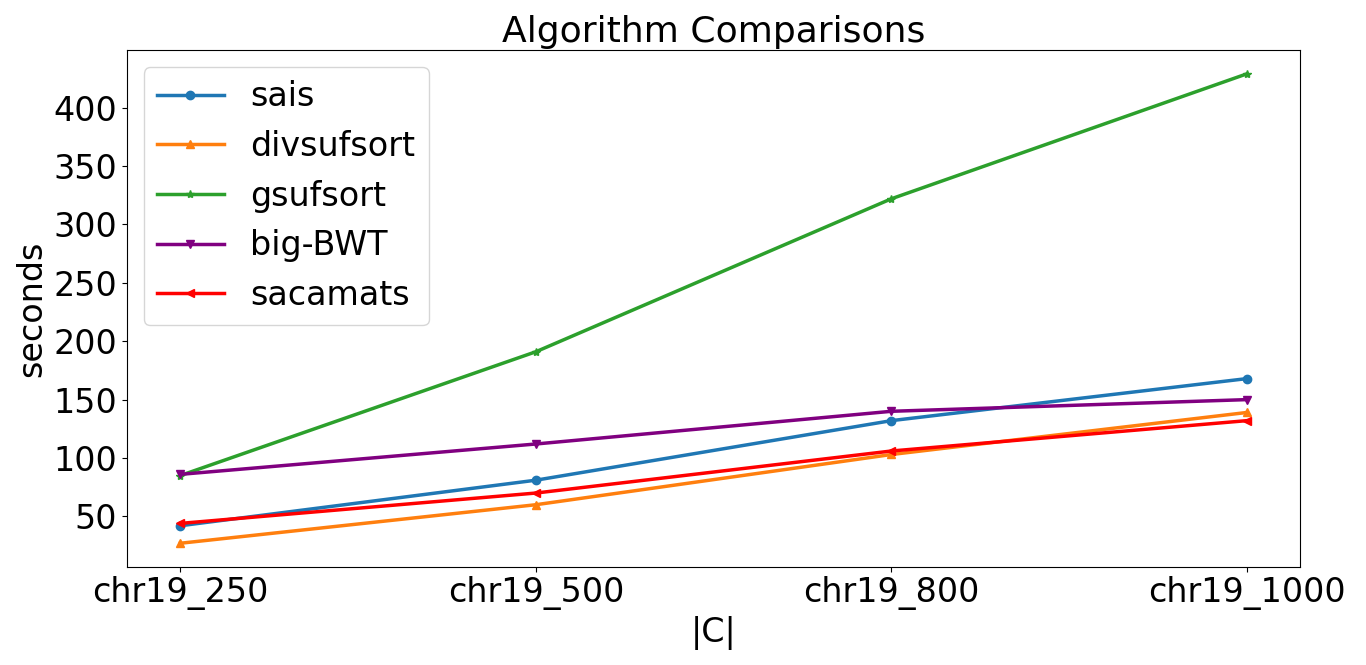}
	\caption{Running time comparison.}
\end{subfigure}%
\begin{subfigure}{.5\textwidth}
    \centering
	\includegraphics[width=\textwidth]{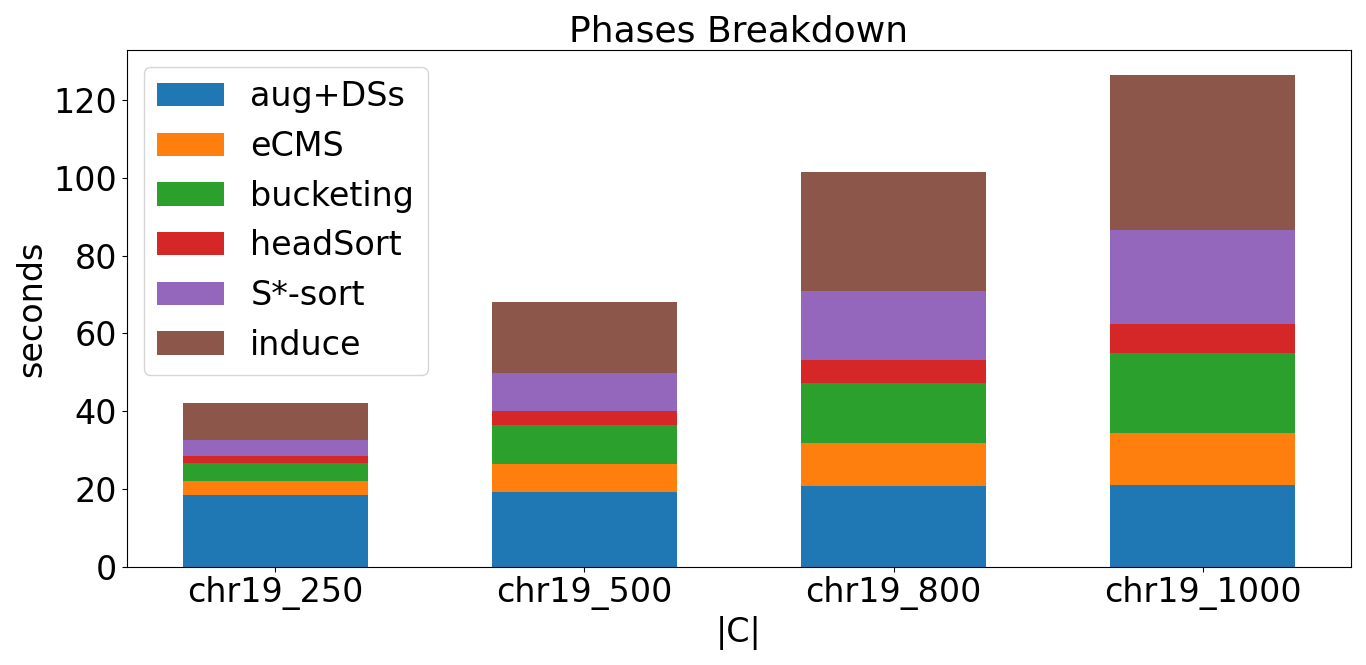}
	\caption{Phases breakdown, see Sec.~\ref{sec:putting} for details.}
	\label{fig:chr19-phases}
\end{subfigure}
\caption{Experiments on different subsets of copies of Human Chromosome 19.}
\label{fig:chr19}
\end{figure}

\begin{figure}[ht]
\centering
\begin{subfigure}{.5\textwidth}
    \centering
	\includegraphics[width=\textwidth]{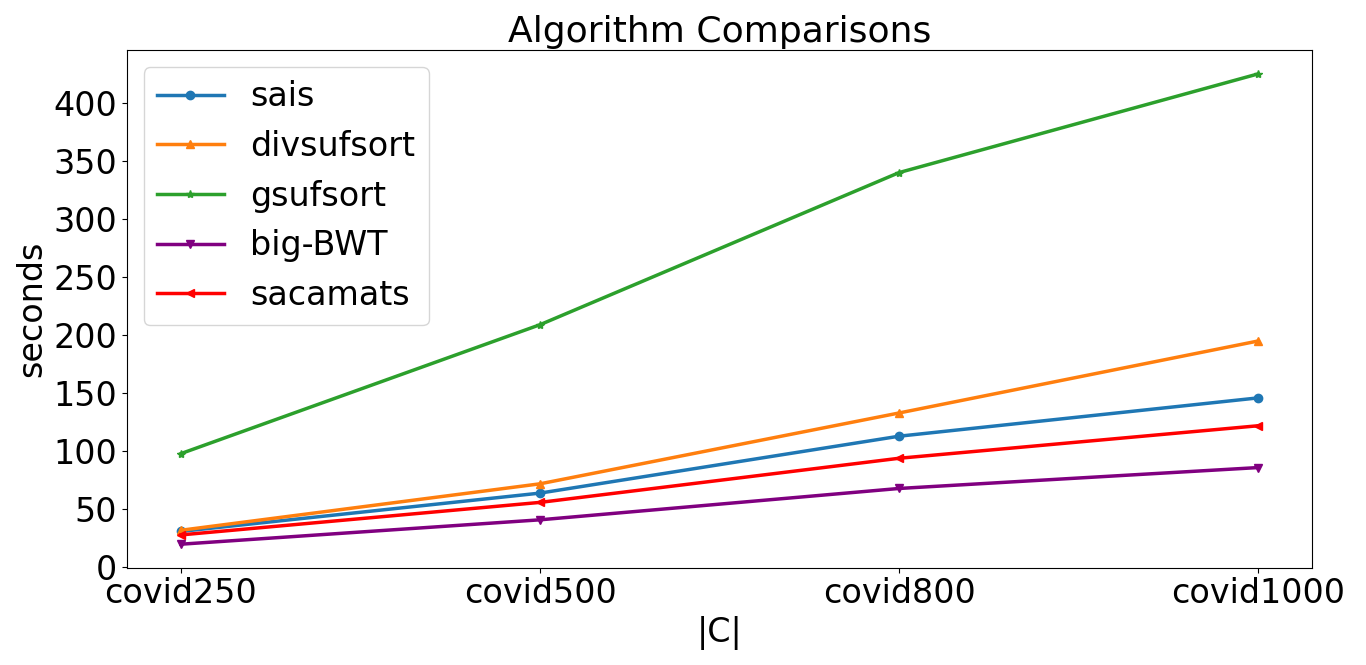}
	\caption{Running time comparison.}
\end{subfigure}%
\begin{subfigure}{.5\textwidth}
    \centering
	\includegraphics[width=\textwidth]{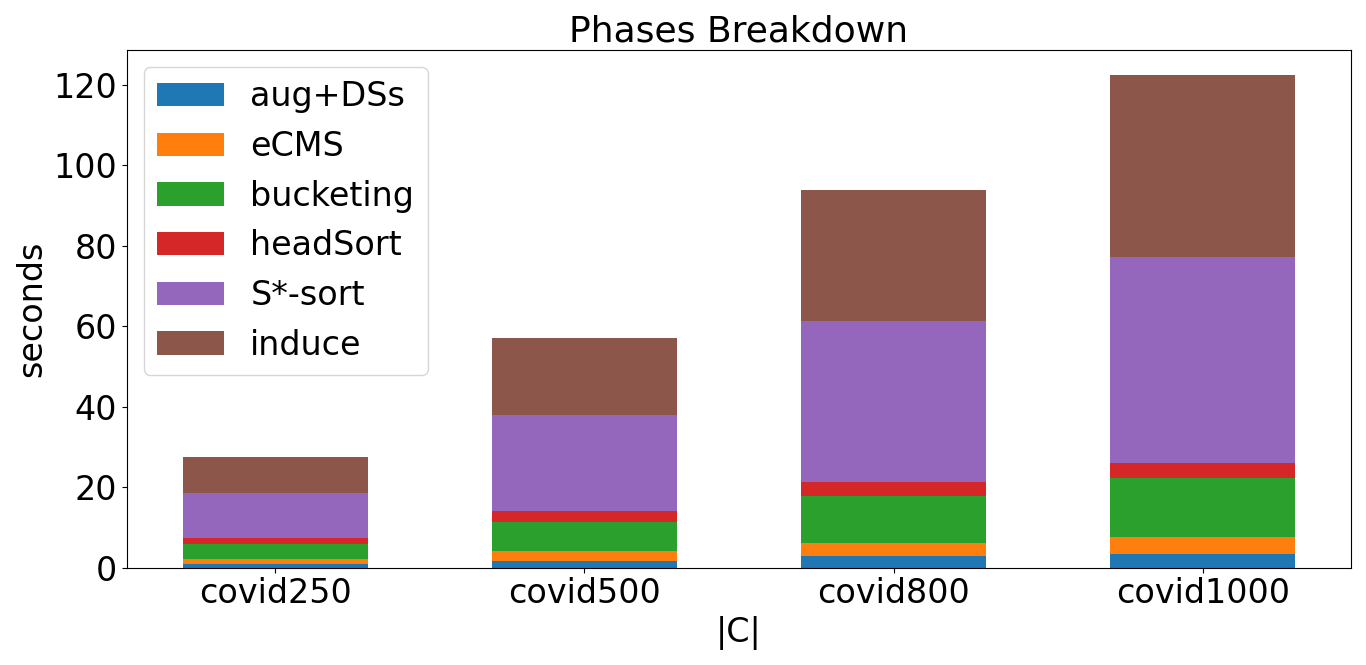}
	\caption{Phases breakdown.}
\end{subfigure}
\caption{Experiments on different subsets of SARS-CoV2 genomes.}
\label{fig:covid}
\end{figure}

%%%%%%%%%%%%%%%%%%%%%%%%%%%%%%%%%%%%%%%%%%%%%%%%%%%%%%%%%%%%%%%%%%%%%%%%%%%%%%%%%%%%%%%
%% SECTION 
%%%%%%%%%%%%%%%%%%%%%%%%%%%%%%%%%%%%%%%%%%%%%%%%%%%%%%%%%%%%%%%%%%%%%%%%%%%%%%%%%%%%%%%

\section{Conclusion}\label{sec:conclusion}

We have presented a new algorithm for computing the generalized suffix array of a collection of highly similar strings. It is based on a compressed representation of the matching statistics, and on efficient handling of string comparisons. Our experiments show that a relatively straightforward implementation of the new algorithm is competitive with the fastest existing suffix array construction algorithms on datasets of highly similar strings, as are common in computational biology applications. 

A byproduct of our suffix sorting algorithm is a heuristic for fast computation of the matching statistics of a collection of highly similar genomes w.r.t. a reference sequence, which is of independent interest. We also envisage uses for our compressed matching statistics (CMS) data structure beyond the present paper, for example as a tool for sparse suffix sorting, or for distributed suffix sorting in which the CMS is distributed to all sorting nodes together with a lexicographic range of the suffixes that each particular node is responsible for sorting. From the CMS alone, each node can extract the positions of its suffixes and then sort them with the aid of the CMS.

We believe there to be a great deal of room for further practical improvements, both through algorithm engineering and parallelism. Interestingly, in an initial attempt along the second line, simply assigning $S^*$-suffix buckets to one of four different sorting threads reduces runtime significantly, for example, from 122 to 89 seconds on the 1GB {\tt sars-cov2} dataset. 

Further studies will be conducted on how the size of \eCMS\ impacts on the competitiveness of our tool.

\end{document}